\documentclass[letterpaper, 10 pt, conference]{ieeeconf}  % Comment this line out if you need a4paper

\IEEEoverridecommandlockouts % This command is only needed if you want to use the \thanks command

\usepackage[pdftex, pdfstartview={FitV}, pdfpagelayout={TwoColumnLeft},bookmarksopen=true,plainpages = false, colorlinks=true, linkcolor=black, citecolor = black, urlcolor = black,filecolor=black , pagebackref=false,hypertexnames=false, plainpages=false, pdfpagelabels ]{hyperref}

\usepackage[final]{changes} %Use this to hide all comments.
\newcommand{\changed}[1]{#1}

\usepackage{setspace}
\usepackage{graphicx}
\usepackage{epstopdf}
\usepackage{amssymb,amsmath}

\usepackage{amsthm}
\usepackage{cuted}
\usepackage[font=footnotesize]{subcaption}

\usepackage[font=footnotesize]{caption}
\usepackage{xcolor}
\usepackage{units}
\usepackage{algorithm}
\usepackage[noend]{algpseudocode}
\usepackage{balance}
\usepackage[sort,compress,noadjust]{cite}
\usepackage{tabularx}
\usepackage{nameref}
\usepackage{centernot}

\let\labelindent\relax
\usepackage{enumitem}

\theoremstyle{plain}
\newtheorem{theorem}{Theorem}
\theoremstyle{plain}
\newtheorem{proposition}{Proposition}
\theoremstyle{plain}
\newtheorem{lemma}{Lemma}
\theoremstyle{plain}

\theoremstyle{definition}

\theoremstyle{definition}
\newtheorem{definition}{Definition}
\theoremstyle{remark}
\newtheorem{remark}{Remark}

\usepackage[hang,flushmargin]{footmisc}

\usepackage[capitalize]{cleveref}
\crefformat{equation}{(#2#1#3)}
\Crefformat{equation}{Equation~(#2#1#3)}
\Crefname{equation}{Equation}{Eqs.}

\usepackage{accents}

\title{\LARGE \bf Universal Adaptive Control of Nonlinear Systems}

\author{Brett T. Lopez$^1$ and Jean-Jacques E. Slotine$^2$% stops a space
\thanks{$^{1}$Versatile Control-Theoretic Robotics Laboratory, University of California -- Los Angeles, Los Angeles, CA, {\tt\small btlopez@ucla.edu}}
\thanks{$^{2}$Nonlinear Systems Laboratory, Massachusetts Institute of Technology, Cambridge, MA, {\tt\small jjs@mit.edu}}}

\begin{document}

\maketitle

\thispagestyle{empty}
\pagestyle{empty}

\begin{abstract}
This work develops a new direct adaptive control framework that extends the certainty equivalence principle to general nonlinear systems with unmatched model uncertainties.
The approach adjusts the rate of adaptation online to eliminate the effects of parameter estimation transients on closed-loop stability. 
The method can be immediately combined with a previously designed or learned feedback policy if a corresponding model-parameterized Lyapunov function or contraction metric is known.
Simulation results of various nonlinear systems with unmatched uncertainties demonstrates the approach.
\end{abstract}

% \begin{IEEEkeywords}

% \end{IEEEkeywords}

\section{Introduction}
\label{sec:introduction}

Concurrent stabilization and model parameter estimation for uncertain nonlinear systems has long been a focus of the controls community. 
Despite several decades of research, a comprehensive direct adaptive control approach has remained elusive.
The main difficulty in developing such a framework is the idea that the certainty equivalence principle cannot be employed when the model uncertainties are outside the span of the control input, i.e., are unmatched.
Departure from certainty equivalence significantly complicates the design process as the controller must either anticipate or be robust to transients in the parameter estimates.
This work will show that the certainty equivalence principle can be extended to systems with unmatched uncertainties by actively adjusting the rate of adaptation.
The approach only requires the uncertainty be linearly parameterized and imposes no restrictions on the type or structure of the dynamics.
% \emph{without} imposing the system be in a particular form.
% This is achieved by accelerates or decelerates the adaptation rate depending on whether the adaptation transients is stabilizing or destabilizing. 
Due to its generality, the approach is referred to as \emph{universal adaptive control}.

Initial work in adaptive control of nonlinear systems used Lyapunov-like stability arguments and the certainty equivalence principle to construct stabilizing adaptive feedback policies for feedback linearized systems with matched uncertainties \cite{slotine1986adaptive,slotine-li1987adaptive,taylor1989adaptive,isidori1989,kanellakopoulos1991extended}, i.e., those that can be directly canceled through control.
% Early approaches were based on feedback linearized models \cite{slotine1986adaptive,slotine-li1987adaptive,taylor1989adaptive,isidori1989,kanellakopoulos1991extended}.
%or variable structure control \cite{ambrosino1984variable}.
% The difficulty of extending early approaches based on feedback linearized models \cite{slotine1986adaptive,slotine-li1987adaptive,taylor1989adaptive,isidori1989,kanellakopoulos1991extended} to more general systems with unmatched uncertainties led to the adaptive backstepping technique \cite{tsinias1991,kanellakopoulos1991systematic,krstic1992adaptive} which relies on the recursive application of the extended matching condition \cite{krstic1995nonlinear} and certainty equivalence \bxx{add 2000 paper}.
% However, adaptive backstepping is limited to systems that have a triangular structure.
The difficulty of extending these early approaches to general systems with unmatched uncertainties led to the development of adaptive backstepping \cite{tsinias1991,kanellakopoulos1991systematic,krstic1992adaptive,tsinias2000backstepping}.
Restricting the system to have a triangular structure enabled the recursive application of the extended matching condition \cite{krstic1995nonlinear} and certainty equivalence.
Combining a control Lyapunov function (clf) with parameter adaptation has also been investigated \cite{krstic1995nonlinear,boffi2021higher} although these methods are often limited to systems with matched uncertainties.
% \changed{In \cite{krstic1995control}, adaptive clf's (aclf) were proposed for systems with unmatched uncertainties, but the existence of an aclf is difficult to establish for general nonlinear systems.}
In \cite{krstic1995control}, adaptive clf's (aclf) were proposed for systems with unmatched uncertainties.
\changed{However, computing an aclf is difficult because it entails stabilizing a modified system whose dynamics depend on its own clf.}
% In \cite{krstic1995control}, adaptive clf (aclf) were proposed for systems with unmatched uncertainties but are difficult to compute as they require finding a clf for a modified system whose dynamics depend on their own clf.
The above challenges spurred interest in indirect methods where an identifier is combined with an input-to-state stable (ISS) controller \cite{krstic1995nonlinear}.
\changed{The strict ISS condition is needed as the identifier may not be fast enough to achieve closed-loop stability -- an issue not encountered with direct methods as they employ Lyapunov-like arguments to derive the adaptation law.}
% The above challenges spurred interest in indirect methods where an identifier is combined with an input-to-state stable (ISS) controller robust to both model uncertainty and transients in the parameter estimates \cite{krstic1995nonlinear}.
% \changed{ISS is required as the identifier may not be fast enough to maintain closed-loop stability -- an issue not encountered with direct methods.} 
Recently, \cite{boffi2021regret} used random basis functions \cite{sanner1991gaussian,boffi2021random} to approximate unmatched uncertainties as matched but can require extensive parameter tuning as the underlying physics are not exploited. 

The main contribution of this work is a new direct adaptive control framework based on certainty equivalence design for general nonlinear systems with unmatched uncertainties. 
The approach is comprised of two core ideas.
The first is defining the \emph{unmatched control Lyapunov function} which is a family of clf's parameterized over all possible models.
% This follows the certainty equivalence philosophy of computing ``infinitely-many" stabilizing clf as opposed to computing a single clf for all models (the latter being a more difficult task).
This follows the certainty equivalence philosophy of computing ``infinitely-many" stabilizing clf's instead of just one for all models.
The second is to adjust the adaptation rate online to eliminate the effects of estimation transients on closed-loop stability.
% \changed{Combining these ideas yields a certainty equivalence adaptive controller that only requires the system remain stabilizable for all possible parametric variations}.
% Combining these ideas yields a certainty equivalence adaptive controller that can leverage the true physics of the system when available.
These ideas are further extended by introducing \emph{unmatched control contraction metrics}, building upon~\cite{lopez2020adaptive}, where a differential~\cite{lohmiller1998contraction,manchester2017control} rather than explicit clf is utilized.
\changed{The approach only requires the system be stabilizable or contracting for every parametric variation -- an intuitive criteria easily included in algorithms that compute a clf or contraction metric.}
It can also be immediately combined with analytic or learned controllers with a known model-parameterized clf or contraction metric.
Simulations illustrate the generality and effectiveness of the approach.

\textit{Notation:} Symmetric positive-definite $n\times n$ matrices are denoted as $\mathcal{S}^n_+$.
Positive and strictly-positive scalars are designated as $\mathbb{R}_+$ and $\mathbb{R}_{>0}$ respectively. 
The shorthand notation of a function $T$ parameterized by a vector $a$ with vector argument $s$ is $T_a(s) := T(s;a)$.
The directional derivative of a smooth matrix $M : \mathbb{R}^n \times \mathbb{R} \rightarrow \mathcal{S}^n_+$ along a vector field $v : \mathbb{R}^n \times \mathbb{R} \rightarrow \mathbb{R}^n$ is $\partial_v M(x,t) = \sum_i^n \partial M / \partial x_i \, {v}_i(x,t)$.
A desired trajectory and control input pair is $(x_d,\,u_d)$. 

%%%%%%%%%%%%%%%%%%%%%%%%%%%%%%%%%%%%%%%%%%%

\section{Problem Formulation}
\label{sec:problem_formulation}
This work addresses control of uncertain dynamical systems of the form
\begin{equation}
    \dot{x} = f(x,t) - \Delta(x,t)^\top \theta + B(x,t) u,
    \label{eq:dyn}
\end{equation}
with state $x \in \mathbb{R}^n$, control input $u \in \mathbb{R}^m$, nominal dynamics $f: \mathbb{R}^n \times \mathbb{R} \rightarrow \mathbb{R}^{n}$, and control input matrix $B: \mathbb{R}^n \times \mathbb{R}  \rightarrow \mathbb{R}^{n\times m}$ with columns $b_i(x,t)$ for $i=1,\dots,m$. 
The uncertain dynamics are a linear combinations of known regression vectors $\Delta: \mathbb{R}^n \times \mathbb{R} \rightarrow \mathbb{R}^{p\times n}$ with rows $\varphi_i(x,t)$ for $i=1,\dots,p$ and unknown parameters $\theta \in \mathbb{R}^p$.
\changed{We assume the dynamics \cref{eq:dyn} are locally Lipschitz uniformly and the state $x$ is measurable.}
Systems with non-parametric or nonlinearly parameterized uncertainties can be converted into a linear weighting of handpicked or learned basis functions.
% The proposed method will also be applicable to systems that are not control affine with some additional complexity; this will be briefly discussed at the end of this note.
A a new adaptive control framework based on certainty equivalence is developed for nonlinear systems in the form of \cref{eq:dyn}.

% The purpose of this work is to develop a new adaptive control framework that extends the uncertainty equivalence property to general nonlinear systems of the form \cref{eq:dyn} with unmatched uncertainties.

% This work will derive an adaptive feedback policy $\kappa_{\hat{\theta}} : \mathbb{R}^n \times \mathbb{R}^p \times \mathbb{R} \rightarrow \mathbb{R}^m$ that stabilizes \cref{eq:dyn} by extending the certainty equivalence principle to systems with unmatched uncerainties 
% \begin{equation}
%     \begin{aligned}
%         u &= \kappa(x,\hat{\theta},\rho) \\
%         \dot{\hat{\theta}} &= \alpha(x,\rho) \\
%         \dot{\rho} &= \omega(x,\hat{\theta},\rho)
%     \end{aligned}
% \end{equation}

% derive control policies $\kappa : \mathbb{R}^n \times \mathbb{R}^n \times \mathbb{R}^p \times \mathbb{R} \rightarrow \mathbb{R}^m$ that is able to ensure the closed-loop system converges to any feasible forward-complete desired trajectory given by $x_d$ and $u_d$.

%%%%%%%%%%%%%%%%%%%%%%%%%%%%%%%%%%%%%%%%%%%

\section{Universal Adaptive Control}

\subsection{Overview}
This section presents the main technical results of this letter. 
First, a new type of clf is defined and an equivalence to stabilizability is established.
The new clf is then used to adaptively stabilize \cref{eq:dyn} using the certainty equivalence principle and online adjustment of the adaptation rate. 
The result is then extended to contracting systems where a new control contraction metric is defined.

\subsection{Unmatched Control Lyapunov Functions}
\label{sub:uclf}

\begin{definition}
\label{def:uclf}
A smooth, positive-definite function $V_\theta : \mathbb{R}^n \times \mathbb{R}^p \times \mathbb{R} \rightarrow \mathbb{R}_+$ is an \emph{unmatched control Lyapunov function} (uclf) if it is radially unbounded in $x$ and for each $\theta \in \mathbb{R}^p$
\begin{equation*}
    \begin{aligned}
        \underset{u \in \mathbb{R}^m}{\mathrm{inf}} \left\{ \frac{\partial V_\theta}{\partial t} + \frac{\partial V_\theta}{\partial x} \left[ f - \Delta^\top \theta + B u \right] \right\} \leq  -Q_\theta(x,t)
    \end{aligned}
\end{equation*}
where $Q_\theta : \mathbb{R}^n \times \mathbb{R}^p \times \mathbb{R} \rightarrow \mathbb{R}_+$ is continuously differentiable, \changed{radially unbounded in $x$,} and positive-definite with $Q_\theta(x,t) = 0 \iff x \neq x_d$.
\end{definition}

\changed{
% The following establishes necessary and sufficient conditions for the existence of an uclf.

\begin{proposition}
\label{prop:stable}
The system \cref{eq:dyn} is stabilizable for each $\theta \in \mathbb{R}^p$ if and only if there exists an uclf.
\end{proposition}

\begin{proof} $ $ \par
\noindent``$\Leftarrow$": Follows from \cite{artstein1983stabilization, sontag1989universal}.\\
``$\Rightarrow$": If \cref{eq:dyn} is stabilizable for each $\theta \in \mathbb{R}^p$, then there exists a controller $u$ and a clf $V_\theta(x,t)$ \cite{artstein1983stabilization} which, without loss of generality, satisfies $\dot{V}_\theta(x,t) \leq - Q_\theta(x,t) \leq 0$ where $Q_\theta(x,t)$ is a uniformly positive-definition function.
Then, by \cref{def:uclf} $V_\theta(x,t)$ must be an uclf.
\end{proof}
\cref{prop:stable} highlights the fundamental difference between an uclf and aclf (\cref{def:aclf}, see Appendix): existence of an uclf is equivalent to the stabilizability of \cref{eq:dyn} for each $\theta \in \mathbb{R}^p$.
This equivalence means known approaches that compute a clf can be trivially extended to construct an uclf by searching over $x$ and $\theta$.
Conversely, existence of an aclf is equivalent to \cref{eq:dyn} being \emph{adaptively stabilizable} \cite{krstic1995control}; a property not easily verified as it entails finding either 1) a stabilizing controller and adaptation law or 2) a stabilizing policy for a modified system that depends on its own (unknown) clf.
% Neither condition is easy to establish aside from systems in strict-feedback form.
As shown next, uclf's play a central role in adaptive control of systems with unmatched uncertainties. }
\vspace{-0.2in}

\begin{theorem}
Consider an uncertain system of the form \cref{eq:dyn}.
If an uclf $V_\theta(x, t)$ exists, then, for any strictly-increasing and uniformly-positive scalar function $\upsilon(\rho)$, the closed-loop system is globally asymptotically stable with the unmatched adaptation law
\begin{subequations}
\label{eq:adapt_law_v}
    \begin{align}
        \dot{\hat{\theta}} & = - \upsilon(\rho) \Gamma \Delta(x,t) \frac{\partial V_{\hat{\theta}}}{\partial {x}}^\top, \label{eq:dot_theta_v} \\
        \dot{\rho} & = - \frac{\upsilon(\rho)}{\upsilon_\rho(\rho)} \sum_{i=1}^p\frac{1}{V_{\hat{\theta}}(x,t) + \eta} \frac{\partial V_{\hat{\theta}}}{\partial \hat{\theta}_i} \dot{\hat{\theta}}_i, \label{eq:dot_rho_V}
    \end{align}
\end{subequations}
where $\Gamma \in \mathcal{S}^p_+$, $\eta \in \mathbb{R}_{>0}$, and $ \upsilon_\rho(\rho):= {\partial \upsilon}/{\partial \rho}$.
\label{thm:univ_v}
\end{theorem}

\begin{proof}
Consider the Lyapunov-like function
\begin{equation*}
    V_c(t) = \upsilon(\rho) (V_{\hat{\theta}}(x,t)+\eta) + \tfrac{1}{2} \tilde{\theta}^\top \Gamma^{-1} \tilde{\theta},
\end{equation*}
where \changed{$0 < \eta < \infty$}, and $\tilde{\theta}:= \hat{\theta} - \theta$.
Differentiating $V_c(t)$,
\begin{equation*}
    \begin{aligned}
        \dot{V}_c(t) & = \upsilon(\rho) \dot{V}_{\hat{\theta}}(x,t) +  \dot{\rho} \upsilon_{\rho}(\rho) (V_{\hat{\theta}}(x,t)+\eta) + \tilde{\theta}^\top \Gamma^{-1} \dot{\hat{\theta}}  \\
        & =  \upsilon(\rho) \left\{ \frac{\partial V}{\partial t} + \frac{\partial V_{\hat{\theta}}}{\partial x} \left[ f - \Delta^\top \hat{\theta} + Bu \right] \right\}  \\
        & \hphantom{=} + \upsilon(\rho) \frac{\partial V_{\hat{\theta}}}{\partial x} \Delta(x,t)^\top \tilde{\theta} + \upsilon(\rho) \sum_{i=1}^p \frac{\partial V_{\hat{\theta}}}{\partial \hat{\theta}_i} \dot{\hat{\theta}}_i \\
        & \hphantom{=} +  \dot{\rho} \upsilon_{\rho}(\rho) (V_{\hat{\theta}}(x,t)+\eta) + \tilde{\theta}^\top \Gamma^{-1} \dot{\hat{\theta}},
        % \leq &~ - \upsilon(\rho) Q_{\hat{\theta}}(x,t)  + \upsilon(\rho) \frac{\partial V_{\hat{\theta}}}{\partial x} \Delta(x,t)^\top \tilde{\theta} + \upsilon(\rho) \sum_{i=1}^p \frac{\partial V_{\hat{\theta}}}{\partial \hat{\theta}_i} \dot{\hat{\theta}}_i \\
        % & ~ +  \dot{\rho} \upsilon_{\rho}(\rho) (V_{\hat{\theta}}(x,t)+\eta) + \tilde{\theta}^\top \Gamma^{-1} \dot{\hat{\theta}},
    \end{aligned}
\end{equation*}
where $\upsilon_{\rho}(\rho)= \partial \upsilon / \partial \rho$. 
Employing the certainty equivalence property in conjunction with \cref{def:uclf} and~\eqref{eq:adapt_law_v} yields $\, \dot{V}_c(t) \leq - \upsilon(\rho) Q_{\hat{\theta}}(x,t) \leq 0 \,$
% \begin{equation*}
%     \begin{aligned}
%         \dot{V}_c(t) \leq & -  \upsilon(\rho) Q_{\hat{\theta}}(x,t)  + \upsilon(\rho) \sum_{i=1}^p \frac{\partial V_{\hat{\theta}}}{\partial \hat{\theta}_i} \dot{\hat{\theta}}_i \\
%         & + \dot{\rho} \upsilon_{\rho}(\rho) (V_{\hat{\theta}}(x,t)+\eta).
%     \end{aligned}
% \end{equation*}
% Using~\eqref{eq:dot_rho_V} in turn yields $\, \dot{V}_c(t) \leq - \upsilon(\rho) Q_{\hat{\theta}}(x,t) \leq 0 \,$
%\begin{equation*}
%    \dot{\rho} = - \frac{1}{2} \frac{\upsilon(\rho)}{\upsilon_{\rho}(\rho)} \sum_{i=1}^p\frac{1}{E} \frac{\partial E}{\partial \hat{\theta}_i} \dot{\hat{\theta}}_i,
%\end{equation*}
%leads to
% \begin{equation}\label{eq:v_c_dot_v}
%     dot{V}_c(t) \leq - \upsilon(\rho) Q_{\hat{\theta}}(x,t) \leq 0
% \end{equation}
which implies that both $\, \upsilon(\rho) (V_{\hat{\theta}}(x,t)+\eta) \, $ and $\, \tilde{\theta} \,$ are bounded.
Since $\, V_{\hat{\theta}}(x,t) > 0 \, $ for all $\, x \neq x_d \, $ and $\, \upsilon(\rho) > 0$ uniformly, then both $\, V_{\hat{\theta}}(x,t)\, $ and $\, \upsilon(\rho)\, $ are bounded for all $\, x \neq x_d$.
\changed{Since $\, \smash{\dot{\hat{\theta}}} = 0\, $ when $x = x_d$, then from \cref{eq:dot_rho_V} $\, \dot{\rho} = 0\, $ so $\, \upsilon(\rho)\, $ remains bounded.
% Noting $\, \upsilon(\rho)\,$ is lower-bounded and $\,\dot{\rho} > 0 \iff \sum_{i}^p \partial V_{\hat{\theta}} / \partial \hat{\theta}_i \, \smash{ \dot{\hat{\theta}}_i} < 0$, i.e., the adaptation transients is a stabilizing term, then setting $\dot{\rho} = 0$ does not affect stability and ensures $\,\upsilon(\rho) < \infty$.
Hence, $\, V_{\hat{\theta}}(x,t)\, $ is bounded because $\, \eta$, $\, \upsilon(\rho)$, and $\, \upsilon(\rho) (V_{\hat{\theta}}(x,t) + \eta) \, $ are bounded}.
% \changed{Since $\, \upsilon(\rho)\, $ is initially bounded, and noting $\smash{\dot{\hat{\theta}}} = 0$ when $x = x_d$, then $\, \upsilon(\rho)\, $ remains bounded.
% % Noting $\, \upsilon(\rho)\,$ is lower-bounded and $\,\dot{\rho} > 0 \iff \sum_{i}^p \partial V_{\hat{\theta}} / \partial \hat{\theta}_i \, \smash{ \dot{\hat{\theta}}_i} < 0$, i.e., the adaptation transients is a stabilizing term, then setting $\dot{\rho} = 0$ does not affect stability and ensures $\,\upsilon(\rho) < \infty$.
% Hence $\, V_{\hat{\theta}}(x,t)\, $ is bounded because $\, \eta$, $\, \upsilon(\rho)$, and $\, \upsilon(\rho) (V_{\hat{\theta}}(x,t) + \eta) \, $ are bounded}.
Differentiating $\, \upsilon(\rho) Q_{\hat{\theta}}(x,t) \,$ and utilizing \cref{eq:dot_rho_V}, 
\begin{equation*}
    \begin{aligned}
        \frac{d}{dt} \left(\upsilon(\rho) Q_{\hat{\theta}}(x,t)\right)  
        = &  ~ \upsilon(\rho) \sum_i^p \left[ \frac{\partial Q_{\hat{\theta}}}{\partial \hat{\theta}_i} - \frac{Q_{\hat{\theta}}(x,t)}{V_{\hat{\theta}}(x,t) + \eta} \frac{\partial V_{\hat{\theta}}}{\partial \hat{\theta}_i}\right] \dot{\hat{\theta}}_i \\
        & + \upsilon(\rho)  \frac{\partial Q_{\hat{\theta}}}{\partial x} \dot{x} + \upsilon(\rho) \frac{\partial Q}{\partial t}, \\
    \end{aligned}
\end{equation*}
which is bounded by continuity of $\, V_{\hat{\theta}}(x,t),~ Q_{\hat{\theta}}(x,t)\,$ and boundedness of $x$ and $\upsilon(\rho)$.
Hence, $\upsilon(\rho) Q_{\hat{\theta}}(x,t) \, $ is uniformly continuous.
Integrating $\dot{V}_c(t)$ yields $\int \limits_0^\infty \upsilon( \rho(\tau) ) Q_{\hat{\theta}}(x(\tau),\tau) d \tau \leq V_c(0) < \infty$, 
% \begin{equation*}
% \end{equation*}
so by Barbalat's lemma~\cite{slotine1991applied} $\, \upsilon (\rho) Q_{\hat{\theta}}(x,t) \rightarrow 0$.
Since $\, \upsilon (\rho) > 0 \,$ uniformly and $\, Q_{\hat{\theta}}(x,t) = 0 \iff x = x_d \, $  then $ \, x \rightarrow x_d \, $ as $t\rightarrow +\infty$.
% Since $Q$ is continuously differentiable then $\langle \nabla Q(x),~ f(x) - \Delta^\top \tilde{\theta} + B(x) u \rangle$ is well-defined and bounded by boundedness of $x$ and $\tilde{\theta}$ so $Q$ is uniformly continuous.
% Moreover, integrating \cref{eq:v_c_dot_v},
% \begin{equation}
%     \int \limits_0^\infty Q(x(\tau),\hat{\theta}) dt \leq V_c(0).
% \end{equation}
% so by Barbalat's lemma~\cite{slotine1991applied} $Q \rightarrow 0$. 
% Since $Q = 0 \iff x = x_d$ then $ \ x \rightarrow x_d \ $ as $t\rightarrow +\infty$.
\end{proof}

\begin{remark} 
\cref{thm:univ_v} immediately extends to the case when $\sqrt{Q_\theta(x,t)}$ is an input to a virtual contracting system~\cite{lohmiller1998contraction,wang2005partial} with $x$ and $x_d$ as particular solutions, rather than satisfying $Q_\theta(x,t) = 0 \Longleftrightarrow x = x_d$.
As in~\cite{slotine2003modular}, this follows from the hierarchical combination property of contracting systems, which generalizes the notion of a sliding variable~\cite{slotine1991applied}.
\end{remark}

\begin{remark}
If the unknown parameters belong to a closed convex set $\Theta$, i.e., $\theta \in \Theta \subset \mathbb{R}^p$, then one can employ the projection operator $\mathrm{Proj}_\Theta (\cdot)$ to ensure $\hat{\theta} \in \Theta$ without affecting stability \cite{slotine1991applied,ioannou2012robust}.
% If the unknown parameters belong to a closed convex set $\Theta$, i.e., $\theta \in \Theta \subset \mathbb{R}^p$, then one can employ the projection operator $\mathrm{Proj}_\Theta (\cdot)$ to ensure $\hat{\theta} \in \Theta$.
% It can be shown that $\dot{V}_c(t) \leq 0$ holds even when $\hat{\theta}$ is at the boundary of set $\Theta$, as discussed in \cite{slotine1991applied,ioannou2012robust}.
\end{remark}

\begin{remark}
If an uclf is designed to use full desired trajectory for feedback, then an \emph{adaptive reference model} -- where $(x_d,\,u_d)$ is re-computed for every new parameter estimate $\hat{\theta}$ -- is required for closed-loop stability.
The origin of adaptive reference models will be discussed further in \cref{sub:uccm}. 
\end{remark}

\cref{thm:univ_v} shows how the certainty equivalence property can be extended to systems with unmatched uncertainties by combining uclf with an \emph{effective} adaptation gain $\upsilon(\rho)\Gamma$ that adjusts in response to whether the parameter adaptation transients terms are stabilizing or destabilizing in $\dot{V}_c(t)$.
Inspecting \cref{eq:dot_rho_V}, if the parameter adaptation transients is destabilizing  then $\upsilon(\rho)\Gamma$ decreases thereby slowing the rate of parameter adaptation; the opposite occurs when the adaptation transients is stabilizing.
% When the parameter adaptation transients are stabilizing then the rate of parameter adaptation accelerates.
Note the effective adaptation gain can be kept constant in this scenario, i.e., set $\dot{\rho} = 0$, without affecting stability.
This is advantageous given the well-known negative effects of high-rate adaptation.
% The above interpretation can be inferred from definition of $\upsilon(\rho)$ and  where the sign of $\dot{\rho}$ is governed by the sign of $\sum_{i}^p{\partial V_{\hat{\theta}}}/{\partial \hat{\theta}_i} \, \smash{\dot{\hat{\theta}}_i}$.

% Not all uclf require adaptive reference models, e.g., output tracking for a sufficiently smooth reference $y_d$, but as will be discussed further in \cref{sub:uccm}.
% Unlike standard model reference adaptive control, the presence of unmatched uncertainties require adaptation of the reference model when a full state trajectory is need for control since the controller is unable to directly cancel any difference between the actual and reference model.
% The genesis of adaptive reference models is more clear within the contraction framework and will be discussed further in \cref{sub:uccm}. 

\changed{Sufficiency for \cref{thm:univ_v} requires the existence of an uclf, which is guaranteed if \cref{eq:dyn} is stabilizable for each $\theta$ (\cref{prop:stable}).}
% Sufficiency for \cref{thm:univ_v} requires the existence of a function $V_\theta(x,t)$ that meets the conditions of \cref{def:uclf}.
% \cref{prop:stable} shows that an uclf exists so long as \cref{eq:dyn} is stabilizable for each $\theta$.
% \cref{prop:stable} showed that such a function is guaranteed to exist if and only if \cref{eq:dyn} is stabilizable for each $\theta$.
Practically, methods like backstepping, feedback linearization, or sum-of-squares (SOS) optimization can be utilized to compute $V_\theta(x,t)$ analytically while more recent data-driven approaches \cite{giesl2020approximation,boffi2020learning} can also be used.
For linear systems one can solve a parameter-dependent algebraic Ricatti equation.
The above approaches find an explicit state transformation $z_\theta = T_\theta(x)$ where $V_\theta(x) = \tfrac{1}{2}z_\theta(x)^\top z_\theta(x)$ is a suitable uclf.
\cref{sub:uccm} will show how the unmatched adaptation law \cref{eq:adapt_law_v} can be combined with contraction theory to instead use a differential transformation $\delta_{z} = \mathcal{T}_\theta(x,t)\delta_x$ yielding a more general result.

\subsection{Unmatched Control Contraction Metrics}
\label{sub:uccm}

Contraction analysis~\cite{lohmiller1998contraction} uses differential geometry to construct stabilizing feedback controllers without constructing an explicit state transformation.
This is achieved by deriving a differential controller $\delta_u$ for the nominal differential dynamics of \cref{eq:dyn} given by $\dot{\delta}_x = A(x,u,t) \delta_x + B(x,t) \delta_u$ where $A(x,u,t):= \partial f / \partial x + \sum^m_i \partial b_i / \partial x \, u_i$.
Convex constructive conditions can be formulated for the so-called control contraction metric~\cite{manchester2017control}  $\ M : \mathbb{R}^n \times \mathbb{R} \rightarrow \mathcal{S}^n_+$ that ensures the distance defined by the metric $M$ between any two points, such as the current and desired state, converges exponentially with rate $\lambda$.
More precisely, for a smooth manifold $\mathcal{M}$ and geodesic $\gamma: \left[0~1\right] \times \mathbb{R} \rightarrow \mathcal{M}$ with boundary conditions $\gamma(0,t) = x_d(t)$ and $\gamma(1,t) = x(t)$, the Riemannian energy $E(x,t) := \int_0^1 \gamma_s(s,t) M(\gamma,t) \gamma_s(s,t) ds$ where $\gamma_s(s,t) := \partial \gamma / \partial s$ satisfies $\dot{E}(x,t) \leq - 2\lambda E(x,t)$ yielding $x \rightarrow x_d$ exponentially with rate $\lambda$. 
In order to leverage the versatility of contraction, a new contraction metric suitable for adaptive control with unmatched uncertainties is required.

\begin{definition}
\label{def:uccm}
A uniformly bounded Riemannian metric $M_\theta : \mathbb{R}^n \times \mathbb{R}^p \times \mathbb{R} \rightarrow \mathcal{S}^n_+$ is an \emph{unmatched control contraction metric} (uccm) if, for each $\theta \in \mathbb{R}^p$, the dual metric $W_\theta(x,t) := M_\theta(x,t)^{-1}$ satisfies
\begin{align}
    & B_\perp^\top \left( W_\theta A_\theta^\top + A_\theta W_\theta - \dot{W}_\theta + 2 \lambda W_\theta \right) B_\perp  \preceq 0 \tag{C1} \label{eq:ccm_c1} \\
    & \partial_{b_i}W_\theta - W_\theta \frac{\partial b_i}{\partial x}^\top  -  \frac{\partial b_i}{\partial x} W_\theta = 0,~~ i=1,\dots,m \tag{C2} \label{eq:ccm_c2}
\end{align}
where $A_\theta(x,u,t):= \partial f / \partial x - \sum^p_i \partial \varphi_i / \partial x \, \theta_i + \sum^m_i \partial b_i / \partial x \, u_i$, $B_\perp(x,t)$ is the annihilator matrix of $B(x,t)$, i.e., $B_\perp^\top B = 0$, and $\dot{W}_\theta := \partial W_\theta / \partial t + \partial_{\dot{x}} W_\theta$.
\end{definition}

\vspace{-0.05in}
\changed{
    \begin{remark}
        The existence of uccm is equivalency to \cref{eq:dyn} be contracting for each $\theta \in \mathbb{R}^p$; a criteria easily included in numerical methods that search for ccm's.
        The proof can be found in \nameref{sec:appendix}.
    \end{remark}
} \vspace{-0.05in}

\begin{remark}
Since the proposed approach employs the the certainty equivalence principle, $\dot{W}_\theta$ does not include a parameter adaptation term which makes computing uccm tractable.
\end{remark}

\begin{remark}
If the uncertainty is matched then the metric does not need to depend on the unknown parameters \cite{lopez2020adaptive}.
\end{remark}

An important property of uccm's is summarized in \cref{lemma:uccm}; it will be later used in \cref{thm:univ_ccm}.

\begin{lemma}
\label{lemma:uccm}
If an uccm $M_\theta(x,t)$ exists, then the Riemannian energy satisfies $\dot{E}_\theta (x,t) \leq - 2 \lambda E_\theta(x,t)$ for each $\theta \in \mathbb{R}^p$.
\end{lemma}

% \begin{proof}
% See Appendix. 
% \end{proof}

\begin{theorem}
\label{thm:univ_ccm}
Consider an uncertain system of the form \cref{eq:dyn}.
If an uccm $M_\theta(x,t)$ exists, then, for any strictly-increasing and uniformly-positive scalar function $\upsilon(\rho)$, the closed-loop system is globally asymptotically stable with the unmatched adaptation law
\begin{subequations}
\label{eq:adapt_ccm}
    \begin{align}
        \dot{\hat{\theta}} & = - \upsilon(\rho) \Gamma \Delta(x,t) M_{\hat{\theta}}(x,t)\gamma_s(1,t), \label{eq:dot_theta_ccm} \\
        \dot{\rho} & = - \frac{\upsilon(\rho)}{\upsilon_\rho(\rho)} \sum_{i=1}^p\frac{1}{E_{\hat{\theta}}(x,t) + \eta} \frac{\partial E_{\hat{\theta}}}{\partial \hat{\theta}_i} \dot{\hat{\theta}}_i, \label{eq:dot_rho_ccm}
    \end{align}
\end{subequations}
where $\Gamma \in \mathcal{S}^p_+$, $\gamma_s := {\partial \gamma}/{\partial s}$ is the geodesic speed, $\eta \in \mathbb{R}_{>0}$, and $ \upsilon_\rho(\rho):= {\partial \upsilon}/{\partial \rho}$.
\end{theorem}

% \begin{proof}
% See Appendix. 
% \end{proof}

\cref{thm:univ_ccm} shows how the results of \cref{sub:uclf} can be extended to contracting systems where only a differential transform $\delta_z = \mathcal{T}_\theta(x,t)\delta_x$ is possible. 
As mentioned previously, the origin of adaptive reference models, i.e., where the desired trajectory $(x_d,\,u_d)$ are re-computed with $\hat{\theta}$, is more obvious under the lens of contraction.
Taking the first variation of the Riemannian energy with \cref{eq:dyn} yields (time dependency omitted in the dynamics)
\begin{equation*}
\begin{aligned}
        \tfrac{1}{2} \dot{E}_{\hat{\theta}}&(x,t)  = \gamma_s(1,t)^\top M_{\hat{\theta}}(x,t) \underbrace{\left[ f(x) - \Delta(x)^\top \hat{\theta} + B(x)u \right]}_\textrm{Estimated System} \nonumber \\
    & - \gamma_s(0,t)^\top M_{\hat{\theta}}(x_d,t) \underbrace{\left[ f({x}_d) - \Delta(x_d)^\top \hat{\theta} + B(x_d)u_d \right]}_\textrm{Adaptive Reference Model} \nonumber \\
    & + \gamma_s(1,t)^\top M_{\hat{\theta}}(x,t)\Delta(x,t)^\top \tilde{\theta} + \tfrac{1}{2}\sum_{i=1}^p \frac{\partial E_{\hat{\theta}}}{\partial \hat{\theta}_i} \dot{\hat{\theta}}_i + \tfrac{1}{2} \frac{\partial E}{\partial t}\\
        % & \leq -2 \lambda E_{\hat{\theta}} + 2\gamma_s(1,t)^\top M_{\hat{\theta}}(x,t)\Delta(x,t)^\top \tilde{\theta} + \sum_{i=1}^p \frac{\partial E_{\hat{\theta}}}{\partial \hat{\theta}_i} \dot{\hat{\theta}}_i,
\end{aligned}
\end{equation*}
where the second term on the right-hand side must depend on the parameter estimate $\hat{\theta}$ in order to leverage \cref{lemma:uccm} and prove \cref{thm:univ_ccm}.
Although the reference model is changing, the system will still exhibit the desired behavior as at least one state can also be arbitrarily specified.
Note \cref{thm:univ_ccm} can be directly used to extend the related formalism of \cite{tsukamoto2021learning} and \cite{richards2021adaptive} to unmatched uncertainties. 
% A natural question then arises as to whether the system will exhibit the desired behavior if the reference model is changing.

% A noteworthy benefit of merging adaptive control with contraction is that one immediately inherits the combination properties of contracting systems. 
% Inspecting the Riemannian energy of an adapting system, one finds that $E_{\hat{\theta}}(x,t)$ satisfies
% \begin{equation*}
%     \dot{E}_{\hat{\theta}}(x,t) \leq -2 \lambda E_{\hat{\theta}}(x,t) + 2\gamma_s(1,t)^\top M_{\hat{\theta}}(x,t)\Delta(x,t)^\top \tilde{\theta} + \sum_{i=1}^p \frac{\partial E_{\hat{\theta}}}{\partial \hat{\theta}_i} \dot{\hat{\theta}}_i,
% \end{equation*}
% where the parameter error and adaptation transients can be viewed as external inputs of a contracting system. 
% Moreover, these terms must tend toward zero since $x \rightarrow x_d$ as shown in \cref{thm:univ_ccm}.
% Therefore, known results \bxx{cite} of combinations of perturbed contracting systems are immediately applicable.
% This is particularly useful in robotics or biological systems \bxx{cite} where several uncertain dynamical systems can now be combined to achieve a desired global behavior while being locally stabilized through adaptive control. \bxx{Shorten this point.}

%%%%%%%%%%%%%%%%%%%%%%%%%%%%%%%%%%%%%%%%%%%

\begin{figure}[t!]
\centering 
  \subfloat[Norm of state $x$.]{\includegraphics[width=.48\columnwidth]{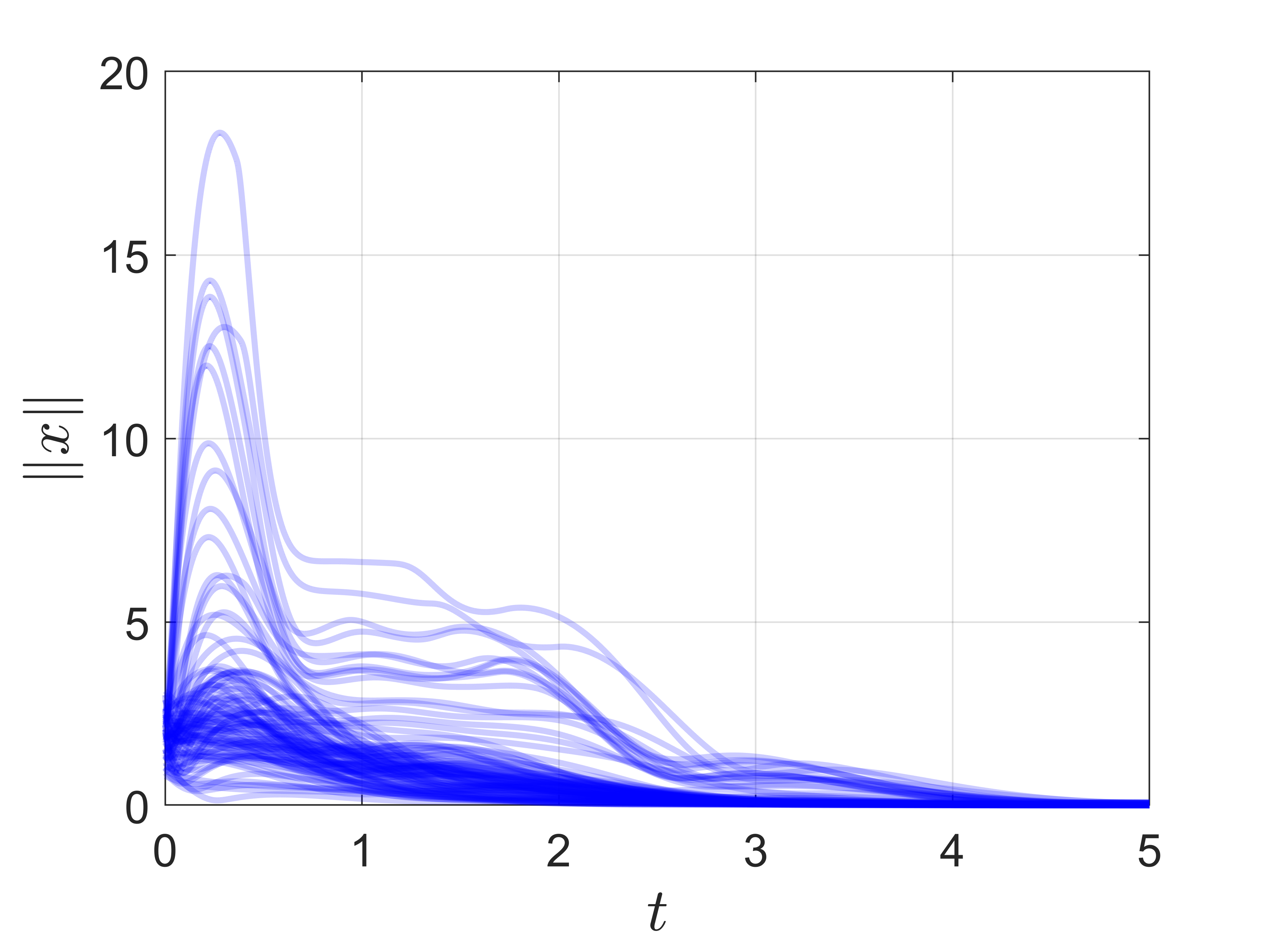}\label{fig:x_ex1}}
  \hspace{0.3em}
  \subfloat[Scaling function $\upsilon(\rho)$.]{\includegraphics[width=.48\columnwidth]{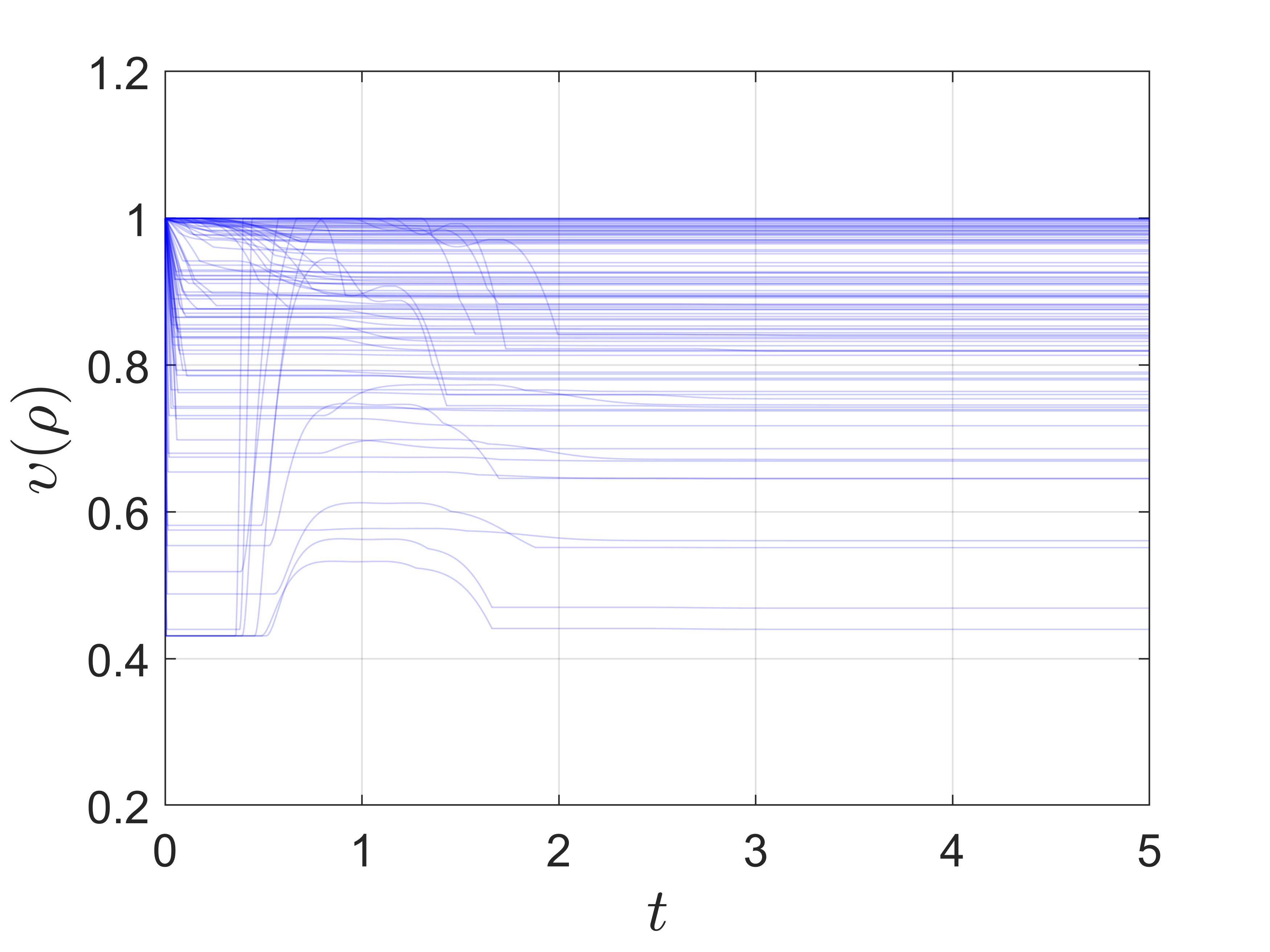}\label{fig:v_ex1}}
  \caption{Norm of state $x$ and the scaling function $\upsilon(\rho)$ for 100 different simulations where $x_0$ and $\theta^*$ are sampled uniformly from bounded sets. (a): Every simulation instance converged to the origin despite the unmatched model uncertainty. (b): Scaling function adjusting to achieve closed-loop stability where $\upsilon(\rho) = 1$ corresponds to no change in the adaptation rate. } 
    \label{fig:ex1}
    \vskip -0.2in
\end{figure}

\section{Numerical Simulations}
\label{sec:results}

\subsection{Example 1: Strict-Feedback System}
Consider the following system with state $x = [x_1~x_2~x_3]^\top$, unknown parameters $\theta = [\theta_1~\theta_2]^\top$, and dynamics
\begin{equation}
\label{eq:ex_dyn_1}
\left[ \begin{array}{c} \dot{x}_1 \\ \dot{x}_2 \\ \dot{x}_3 \end{array} \right] = \left[ \begin{array}{c}  - \theta_1 \mathrm{sin}(x_1) - \theta_2 x_1^2 + x_2 \\ x_3  \\ 0 \end{array} \right] + \left[ \begin{array}{c} 0 \\ 0 \\ 1 \end{array} \right] u,
\end{equation}
which are in strict feedback form with unmatched uncertainties.
\changed{Note that \cref{eq:ex_dyn_1} is a reasonable benchmark since triangular systems are commonly found in the nonlinear adaptive controls literature.
The developed approach is applicable to a much broader class of systems as shown in \cref{prop:stable}.}
The goal is to stabilize the origin where the initial state $x_0$ and true parameters $\theta^*$ are sampled uniformly from bounded sets, i.e., $x_0 \in [-2,\,2]^3$ and $\theta^* \in [-0.2,\,0.4] \times [0.2,\,0.6]$.
The uclf $V_\theta(x) = \tfrac{1}{2}z_\theta(x)^\top z_\theta(x)$ was derived via backstepping\footnote{\changed{The expression for $z_\theta(x)$ and $u_\theta(x)$ can be found in \nameref{sec:appendix}.}} \cite{krstic1995nonlinear} as if $\theta$ was known yielding $\dot{V}_\theta(x) \leq - 4 V_\theta(x)$.
The scaling function was chosen to be $\upsilon(\rho) = 0.9\,e^{\nicefrac{\rho}{5}} + 0.1$ while other parameters were $\Gamma = 0.1I_2$, $\eta = 100$, and simulation time step $dt = 0.005$.

One hundred simulation experiments were conducted with different initial conditions $x_0$ and model parameters $\theta^*$ to evaluate the performance of the proposed adaptive controller.
\cref{fig:x_ex1} shows the norm of the state vector $x$ converges to the origin for each trial as desired; 11\% of the trials without adaptation failed resulting in the closed-loop system diverging.
The scaling function $\upsilon(\rho)$ for each trial is seen in \cref{fig:v_ex1} and shows the adaptation rate adjusting in the majority of trials to maintain stability.  

\subsection{Example 2: Contracting System}

Consider the following system with state $x = [x_1~x_2~x_3]^\top$, unknown parameters $\theta = [\theta_1~\theta_2~\theta_3~\theta_4]^\top$, and dynamics
\begin{equation}
\label{eq:ex_dyn_2}
\left[ \begin{array}{c} \dot{x}_1 \\ \dot{x}_2 \\ \dot{x}_3 \end{array} \right] = \left[ \begin{array}{c}  x_3 - \theta_1 x_1 \\ - x_2 -\theta_2 x^2_1  \\ \mathrm{tanh}(x_2) - \theta_3 x_3 - \theta_4 x_1^2 \end{array} \right] + \left[ \begin{array}{c} 0 \\ 0 \\ 1 \end{array} \right] u.
\end{equation}
The system is not feedback linearizable (controllability matrix drops rank at the origin) and is not in strict feedback form.
Note that $\theta_1$ and $\theta_2$ are unmatched parameters.
The goal is to track a desired trajectory ($x_d,~u_d$) generated with parameters $\theta_{d} = [\hat{\theta}_1\,\hat{\theta}_2\,0\,0]^\top$ driven by a reference $x_{1d}(t) = \mathrm{sin}(t)$.
The true model parameters were $\theta^* = [-0.3\,-0.8\,-0.25\,-0.75]^\top$.
% The set of allowable parameter variations was ${\theta}_1 \in [-0.4,~0.5]$, ${\theta}_2 \in [-1,~0.6]$, ${\theta}_3 \in [-0.6,~0.75]$, and ${\theta}_4 \in [-1.75,~0.4]$; each parameter was bounded by the projection operator from \cite{slotine1991applied}
The set of allowable parameter variations was ${\theta} \in [-0.4,~0.5] \times [-1,~0.6] \times [-0.6,~0.75] \times [-1.75,~0.4]$; each parameter was bounded by the projection operator from \cite{slotine1991applied}.
The scaling function was chosen to be $\upsilon(2\rho) = 0.9\, e^{\nicefrac{\rho}{5}} + 0.1$ while $\Gamma = 5 I_4$, $\eta = 0.1$, and simulation time step $dt = 0.01$.
A dual uccm was computed using YALMIP and SOS \cite{Lofberg2004}; it was non-flat meaning a quadratic clf does not exist.
Geodesics and the pointwise min-norm controller \cite{primbs2000receding} were computed at each time step.

\cref{fig:ex2} compares the performance of the proposed controller to a controller that utilizes the same contraction metric but does not adapt to the unmatched parameters, i.e., the reference mode is fixed.
The Riemannian energy -- an indicator of the closed-loop tracking error -- with and without an adaptive reference model is shown in \cref{fig:E_ex2}.
The tracking performance of the proposed approach is superior to that of the controller with a static reference model demonstrating the importance and effectiveness of adapting the reference model with the current parameter estimates.
Moreover, \cref{fig:x_ex2} confirms that the system will still exhibit the desired behavior (blue curve) despite the reference model changing. 
% In this experiment, the adaptation rate had to be slowed by approximately 1\%.

% The parameter estimates and their true values are shown in \cref{fig:params} while the scalar variable $\rho$ is shown in \cref{fig:rho}.
% Note how initially $\dot{\rho} < 0$ so the adaptation rate is slowed to ensure the closed-loop does not become unstable.
% After the initial transients $\dot{\rho} > 0$ so adaptation is accelerated indicating the model adaptation transients is acting as a stabilizing term.

\begin{figure}[t!]
\centering 
  \subfloat[Riemannian energy $E$ with and without an adapting reference model.]{\includegraphics[trim=120 0 120 0, clip, width=.48\columnwidth]{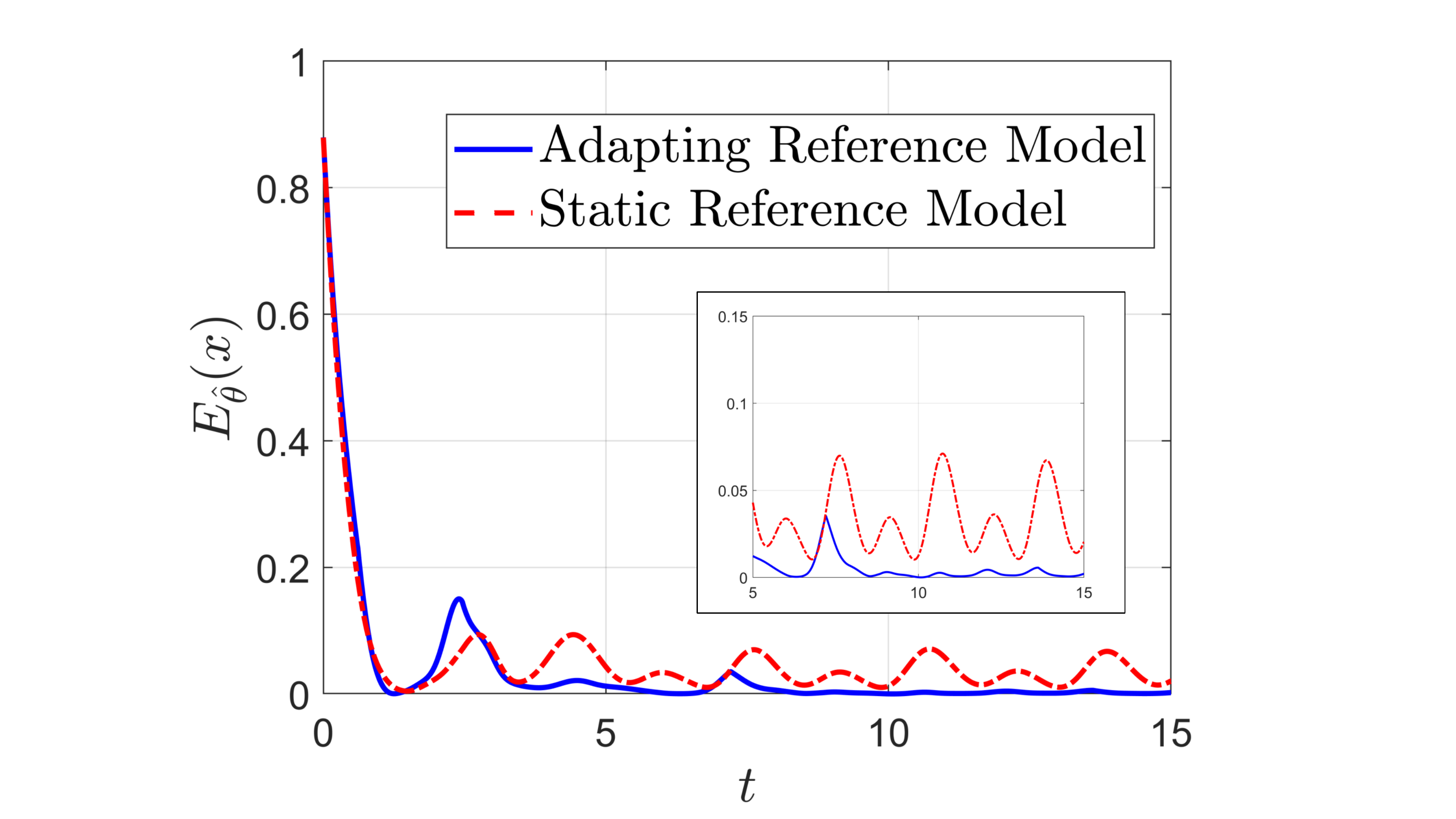}\label{fig:E_ex2}}
  \hspace{0.3em}
  \subfloat[State $x$ with an adapting reference model.]{\includegraphics[width=.48\columnwidth]{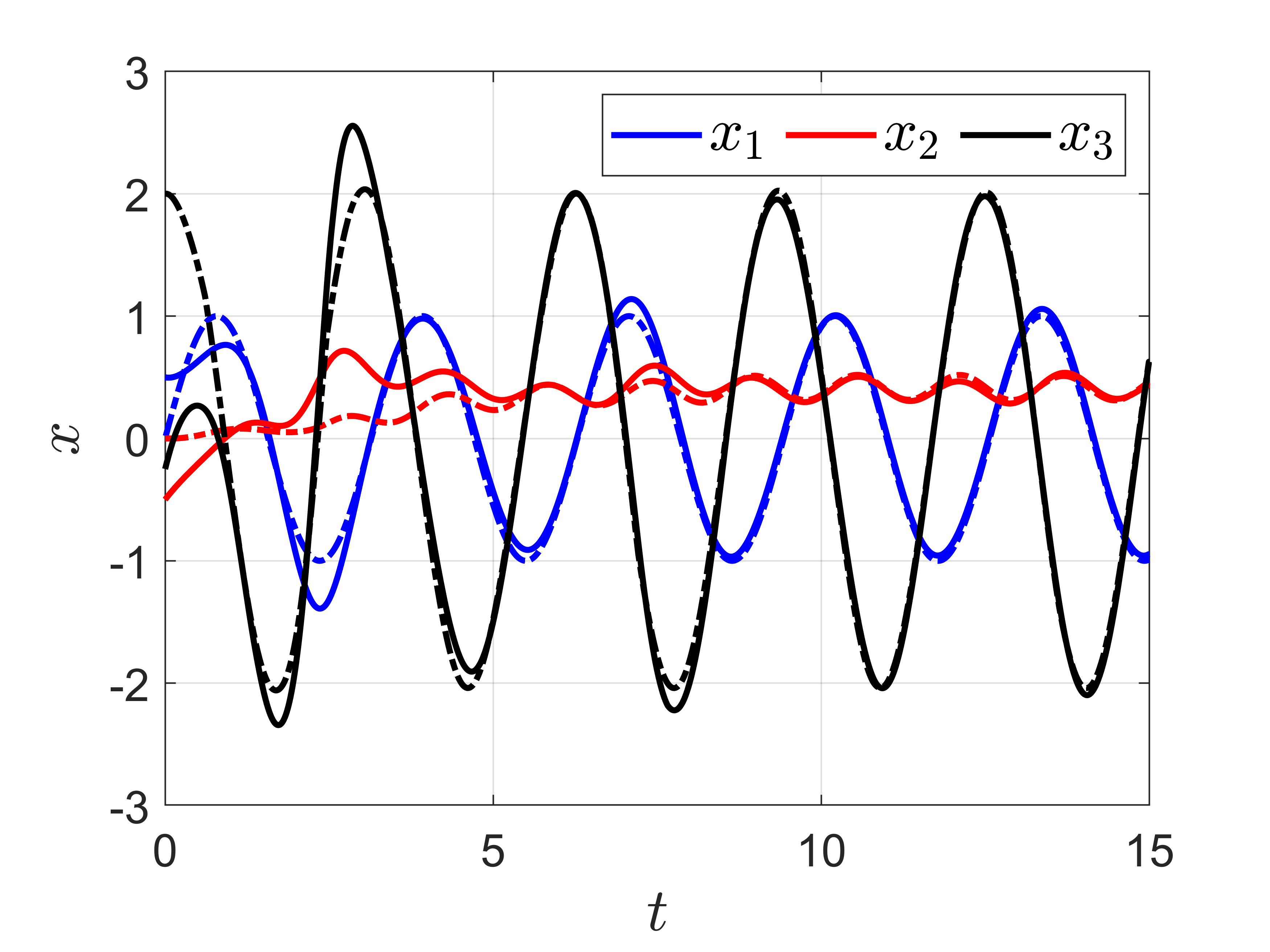}\label{fig:x_ex2}}
  \caption{Controller comparison with and without an adapting reference model. (a): Riemannian energy of the proposed approach is lower than that of a controller with a static reference model. (b): Closed-loop system exhibits desired behavior (blue) despite the reference model adapting online.} 
    \label{fig:ex2}
    \vskip -0.15in
\end{figure}

% \begin{figure}[t!]
% \centering 
%   \subfloat[Parameter estimates.]{\includegraphics[width=.48\columnwidth]{figures/params_2.png}\label{fig:params}}
%   \hspace{0.3em}
%   \subfloat[Scaling function \bxx{change} $\upsilon(\rho)$.]{\includegraphics[width=.48\columnwidth]{figures/rho_2.png}\label{fig:rho}}
%   \caption{Parameter estimates and scaling function $\upsilon(\rho)$. (a): Evolution of parameter estimates as closed-loop system converges to the desired trajectory. (b): The evolution of the scalar variable $\rho$ indicates there are instances where adaption needed to be slowed (corresponding to $\dot{\rho} < 0$) to maintain closed-loop stability. There are also instances where adaptation is accelerated (corresponding to $\dot{\rho} > 0$) indicating the model adaptation transients is actually acting as a stabilizing term.} 
%     \label{fig:adapt}
%     \vskip -0.2in
% \end{figure}

%%%%%%%%%%%%%%%%%%%%%%%%%%%%%%%%%%%%%%%%%%%
\section{Concluding Remarks}
This work developed a new direct adaptive control technique that eliminates the effects of parameter estimation transients on stability through online adjustment of the adaptation rate. 
The key advantage of the method is the ability to leverage the certainty equivalence principle in synthesizing stabilizing controllers via Lyapunov or contraction theory for systems with unmatched uncertainties.
Several extensions are of interest, including combining the unmatched adaptation law with data-driven controllers for high-dimensional nonlinear systems, e.g., robotic manipulation or mobility; expanding the notion of adaptive safety \cite{lopez2020robust} via certainty equivalence design of barrier functions; and the output feedback with adaptive nonlinear observers.
The proposed framework may also be particularly useful in the field of machine learning for dynamical systems, where standard practice is to learn a single feedback policy for all possible model realizations.
With the results presented here, the certainty equivalence principle can now be employed to learn a \emph{family} of optimal policies that then can be later combined with online adaptation to achieve stability despite the presence of uncertainty.
\changed{Finally, further investigation into the connection between adaptive stabilizability \cite{krstic1995control} and stabilizability for each $\theta$ is of interest.}
\section{Appendix}
\label{sec:appendix}
\subsection{Adaptive Control Lyapunov Functions}
\begin{definition}[cf.~\cite{krstic1995control}]
\label{def:aclf}
A smooth, positive-definite function $V_a : \mathbb{R}^n \times \mathbb{R}^p \times \mathbb{R} \rightarrow \mathbb{R}_+$ is an \emph{adaptive control Lyapunov function} (aclf) if it is radially unbounded in $x$ and for each $\theta \in \mathbb{R}^p$ satisfies
\begin{equation*}
    \begin{aligned}
        \underset{u \in \mathbb{R}^m}{\mathrm{inf}} & \left\{ \frac{\partial V_a}{\partial t} + \frac{\partial V_a}{\partial x} \left[ f - \Delta^\top \left( \theta - \Gamma \frac{\partial V_a}{\partial \theta}^\top  \right) + B u \right] \right\} \leq 0.
    \end{aligned}
\end{equation*}
\end{definition}

\subsection{Universal Adaptive Control with uccm}
\begin{proposition}
\label{prop:contract}
The system \cref{eq:dyn} is contracting for each $\theta \in \mathbb{R}^p$ if and only if there exists an uccm.
\end{proposition}

\begin{proof} $ $ \par
\noindent``$\Leftarrow$": Let $\delta V_\theta = \delta_x^\top M_\theta(x,t) \delta_x$ with uccm $M_\theta(x,t)$. Using \cref{def:uccm}, one can show the implication $\delta_x^\top M_\theta B = 0 \Rightarrow \dot{\delta V}_\theta \leq - 2 \lambda \delta V_\theta$ is true for each $\theta \in \mathbb{R}^p$ with uccm $M_\theta(x,t)$. Therefore, \cref{eq:dyn} is contracting for each $\theta \in \mathbb{R}^p$. \\
% If an uccm $M_\theta(x,t)$ exists then the differential energy $\delta V_\theta = \delta_x^\top M_\theta(x,t) \delta_x$ satisfies $\dot{\delta V}_\theta \leq - 2 \lambda \delta V_\theta$ for each $\theta \in \mathbb{R}^p$. Hence, \cref{eq:dyn} is contracting for each $\theta \in \mathbb{R}^p$. \\
``$\Rightarrow$": If \cref{eq:dyn} is contracting for each $\theta \in \mathbb{R}^p$, then there exists a ccm $M_\theta(x,t)$ such that the implication $\delta_x^\top M_\theta B = 0 \Rightarrow \dot{\delta V}_\theta \leq - 2 \lambda \delta V_\theta$ is true for $\delta V_\theta = \delta_x^\top M_\theta(x,t) \delta_x$ and each $\theta \in \mathbb{R}^p$.
Using the constructive conditions from \cite{manchester2017control}, one arrives to those in \cref{def:uccm}. 
Hence, $M_\theta(x,t)$ is an uccm.

% differential energy $\delta V_\theta = \delta_x^\top M_\theta(x,t) \delta_x$ satisfies $\frac{d}{dt}(\delta_x^\top M(x,t) \delta_x) \leq - 2 \lambda \delta_x^\top M(x,t) \delta_x$. If the preceding inequality is expanded then one would obtain the same inequality in \cref{def:uccm}. Therefore, $M_\theta(x,t)$ must be an uccm.
\end{proof}

\begin{proof}[Proof of \cref{lemma:uccm}]
Since an uccm exists then $\, \frac{d}{dt}\left( \delta_x^\top M_\theta(x,t) \delta_x \right) \leq - 2 \lambda \delta_x^\top M_\theta(x,t) \delta_x$.
Integrating along a geodesic $\, \gamma(s,t) \,$ yields $\,\dot{E}_\theta(x,t) \leq - 2 \lambda E_\theta(x,t)$.
\end{proof}

\begin{proof}[Proof of \cref{thm:univ_ccm}]
Proof largely follows that of \cref{thm:univ_v} where $V_{\hat{\theta}}(x,t) = \int_0^1 \gamma_s(s,t)^\top M_{\hat{\theta}}(\gamma,t)\gamma_s(s,t) ds$ and $Q_{\hat{\theta}}(x,t) = 2 \lambda E_{\hat{\theta}}(x,t)$.
Using~\cref{lemma:uccm} and \cref{eq:dot_theta_ccm,eq:dot_rho_ccm}, it can be shown that $\dot{V}_c(t) \leq - \lambda \upsilon(\rho)  E_{\hat{\theta}}(x,t)$ 
which implies that both $\, \upsilon(\rho)( E_{\hat{\theta}}(x,t) + \eta) \,$ and $\, \tilde{\theta} \, $ are bounded.
Using identical arguments to those in \cref{thm:univ_v}, \changed{since $\eta$, $\upsilon(\rho)$, and $\, \upsilon(\rho)( E_{\hat{\theta}}(x,t) + \eta) \,$ are bounded then $E_{\hat{\theta}}(x,t)$ is bounded}.
Differentiating the right side of $\dot{V}_c(t)$,

\begin{equation*}
    \begin{aligned}
        % \frac{d}{dt} ( \upsilon(\rho) E_{\hat{\theta}}(x,t)) = & ~ 2 \upsilon(\rho) \gamma_s(1,t)M_{\hat{\theta}}(x,t) \dot{x} \\
        % & ~ - 2 \upsilon(\rho)\gamma_s(0,t)M_{\hat{\theta}}(x_d,t) \dot{x}_d\\
        \frac{d}{dt} ( \upsilon(\rho) E_{\hat{\theta}}(&x,t)) = ~ 2 \upsilon(\rho) \gamma_s(s,t)M_{\hat{\theta}}(\gamma(s,t),t) \dot{\gamma}(s,t) \Big\rvert^{s=1}_{s=0} \\
        & ~ + \upsilon(\rho) \sum_{i=1}^p \left[ 1 - \frac{E_{\hat{\theta}}(x,t)}{E_{\hat{\theta}}(x,t) + \eta} \right] \frac{\partial E}{\partial \hat{\theta}_i} \dot{\hat{\theta}}_i + \frac{\partial E_{\hat{\theta}}}{\partial t}
        % \frac{d}{dt} ( \upsilon(\rho) E_{\hat{\theta}}(x,t)) & \leq 4  \lambda^2 \upsilon(\rho) E_{\hat{\theta}} - 4 \lambda \upsilon(\rho)  \gamma_s(1,t)^\top M_{\hat{\theta}}(x,t) \Delta(x,t)^\top \tilde{\theta} - 2 \lambda  \upsilon(\rho) \sum_{i=1}^p \frac{\partial E}{\partial \hat{\theta}_i} \dot{\hat{\theta}}_i - 4 \lambda \dot{\rho} \upsilon_{\rho}(2\rho) E_{\hat{\theta}} \\
        % & \leq 4  \lambda^2 \upsilon(\rho)  E_{\hat{\theta}} - 4  \lambda \upsilon(\rho)\gamma_s(1,t)^\top M_{\hat{\theta}}(x,t) \Delta(x,t)^\top \tilde{\theta} - 2 \lambda  \upsilon(\rho) \left(1-\frac{E_{\hat{\theta}}}{E_{\hat{\theta}}+\eta}\right) \sum_{i=1}^p \frac{\partial E_{\hat{\theta}}}{\partial \hat{\theta}_i} \dot{\hat{\theta}}_i, 
    \end{aligned}
\end{equation*}
which is bounded. 
First note that $E_{\hat{\theta}}(x,t),~\tilde{\theta},$ and $\upsilon(\rho)$ are all bounded.
Since $E_{\hat{\theta}}(x,t)$ is bounded then $x$ is also bounded for bounded $x_d$.
Additionally, $\gamma_s(s,t)$ is bounded because geodesics have constant speed so $E_{\hat{\theta}}(x,t) = \langle \gamma_s,\gamma_s\rangle_{M_{\hat{\theta}}}$ and since $E_{\hat{\theta}}(x,t)$ is bounded and $M_{\hat{\theta}}(x,t) \succ 0$ then $\gamma_s(s,t)$ must also be bounded for all $s\in[0~1]$ and $t$.
All terms in \cref{eq:dot_theta_ccm} are bounded so $\smash{\dot{\hat{\small{\theta}}}}$ is also bounded. 
By smoothness ${\partial E_{\hat{\theta}} / \partial \hat{\theta}_i}$ is bounded.
Hence $\, \upsilon(\rho) E_{\hat{\theta}}(x,t) \, $ is uniformly continuous.
Integrating $\dot{V}_c(t)$ yields $\int \limits_0^\infty \upsilon( \rho(\tau) ) E_{\hat{\theta}}(x(\tau),\tau) d \tau \leq V_c(0) \leq \infty$
% \begin{equation}
%     \int \limits_0^\infty \upsilon( \rho(\tau) ) E_{\hat{\theta}}(x(\tau),\tau) d \tau \leq V_c(0) \leq \infty, 
% \end{equation}
so $\, \upsilon (2\rho) E_{\hat{\theta}} (x,t) \rightarrow 0$ by Barbalat's lemma~\cite{slotine1991applied}.
Since $\, \upsilon (\rho) > 0 \, $ uniformly and $\, E_{\hat{\theta}}(x,t) = 0 \iff x = x_d\, $ then $ \, x \rightarrow x_d \, $ as $t\rightarrow +\infty$. 
% Since $\upsilon(2\rho) > 0$ uniformly then by Barbalat's lemma~\cite{slotine1991applied} $\ E_{\hat{\theta}} \rightarrow \ 0$ so $ \ x \rightarrow x_d \ $ as $t\rightarrow +\infty$.
\end{proof}

\subsection{Example 1: Backstepping Controller Derivation}
The state transformation $z_\theta(x)$ for \cref{eq:ex_dyn_1} was computed using the standard backstepping technique, which yields
\begin{equation*}
\begin{aligned}
    z_{\theta,1}(x) & = x_1 \\
    z_{\theta,2}(x) & = x_2 + 2 x_1 - \theta_2 x_1^2 - \theta_1 \mathrm{sin}(x_1) \\
    z_{\theta,3}(x) & = x_1 + x_3 + 2(-\theta_2 x_2^2 + x_2 + x_1 - \theta_1 \mathrm{sin}(x_1)) \\
                    & + (\theta_2 x_1^2 - x_2 + \theta_1 \mathrm{sin}(x_1))(2\theta_2 x_1 - 2 + \theta_1\mathrm{cos}(x_1)).
\end{aligned}
\end{equation*}

\noindent With the controller \cref{eq:u_back}, the uclf $V_\theta(x) = \tfrac{1}{2}z_\theta(x)^\top z_\theta(x)$ satisfies $\dot{V}_\theta(x) \leq - 4 V_\theta(x)$ for each $\theta \in \mathbb{R}^2$.

\begin{table*}[t!]
\centering
\begin{minipage}{\textwidth}
\hrule
\begin{align}
 u_\theta(x) = & \,\,(\theta_2x_1^2 - x_2 + \theta_1 \mathrm{sin}(x_1))\left[ (2\theta_2 x_1 + \theta_1 \mathrm{cos}(x_1)) (2 \theta_2 x_1 - 2 + \theta_1 \mathrm{cos}(x_1)) + (2 \theta_2 - \theta_1 \mathrm{sin}(x_1)) \right. \nonumber \\
 & \left. (\theta_2 x_1^2 - x_2 + \theta_1 \mathrm{sin}(x_1)) - 2 (2 \theta_2 x1 - 2 + \theta_1 \mathrm{cos}(x_1)) - 1) \right]  - 2 \left[x_1 + x_3 + 2 (- \theta_2 x_2^2 + x_2 + 2 x_1 - \theta_1 \mathrm{sin}(x_1)) \right. \nonumber \\
 & + \left. (\theta_2 x_1^2 - x_2 + \theta_1 \mathrm{sin}(x_1)) (2 \theta_2 x_1 - 2 + \theta_1 \mathrm{cos}(x_1))\right] - x_2 - 2 x_1 + \theta_2 x_1^2 + \theta_1 \mathrm{sin}(x_1) \nonumber \\
 & + x_3 (2 \theta_2 x_1 - 4 + \theta_1 \mathrm{cos}(x_1)) \label{eq:u_back}
\end{align}
\medskip
\hrule
\end{minipage}
\end{table*}

%(a2*x1^2 - x2 + a1*sin(x1))*((2*a2*x1 + a1*cos(x1))*(2*a2*x1 - c + a1*cos(x1)) + (2*a2 - a1*sin(x1))*(a2*x1^2 - x2 + a1*sin(x1)) - c*(2*a2*x1 - c + a1*cos(x1)) - 1) - c*(x1 + x3 - x1d + c*(- a2*x2^2 + x2 + c*(x1 - x1d) - a1*sin(x1)) + (a2*x1^2 - x2 + a1*sin(x1))*(2*a2*x1 - c + a1*cos(x1))) - x2 - c*(x1 - x1d) + a2*x1^2 + a1*sin(x1) + x3*(2*a2*x1 - 2*c + a1*cos(x1))

%z =
%                                                                                                               x1 - x1d
%                                                                               x2 + c*(x1 - x1d) - a2*x1^2 - a1*sin(x1)
%x1 + x3 - x1d + c*(- a2*x2^2 + x2 + c*(x1 - x1d) - a1*sin(x1)) + (a2*x1^2 - x2 + a1*sin(x1))*(2*a2*x1 - c + a1*cos(x1))

\noindent \textbf{Acknowledgements:} The authors thank Sumeet Singh for his specific suggestions on an early draft of the manuscript.

% \balance
\bibliographystyle{ieeetr}
\bibliography{ref}

\end{document}